\documentclass[conference]{IEEEtran}
\IEEEoverridecommandlockouts

\usepackage{cite}
\usepackage{amsmath,amssymb,amsfonts}
\usepackage{amsthm}
\usepackage{amsmath}
\usepackage{algorithm}
\usepackage{algorithmic}
\newtheorem{lemma}{Lemma}

\newtheorem{corollary}{Corollary}
\usepackage{graphicx}
\usepackage{array}
\usepackage{textcomp}
\usepackage{xcolor}

\def\BibTeX{{\rm B\kern-.05em{\sc i\kern-.025em b}\kern-.08em
    T\kern-.1667em\lower.7ex\hbox{E}\kern-.125emX}}
\begin{document}

\title{Optimal On-Off Transmission Schemes for Full Duplex Wireless Powered Communication Networks\\
\thanks{This work is supported by Scientific and Technological Research Council of Turkey Grant $\#$117E241.}}
\author{\IEEEauthorblockN{Muhammad Shahid Iqbal}
\IEEEauthorblockA{\textit{Electrical and Electronics Engineering} \\
\textit{Koc University}\\
Istanbul, Turkey \\
miqbal16@ku.edu.tr}
\and
\IEEEauthorblockN{Yalcin Sadi}
\IEEEauthorblockA{\textit{Electrical and Electronics Engineering} \\
\textit{Kadir Has University}\\
Istanbul, Turkey \\
yalcin.sadi@khas.edu.tr }
\and
\IEEEauthorblockN{Sinem Coleri}
\IEEEauthorblockA{\textit{Electrical and Electronics Engineering} \\
\textit{Koc University}\\
Istanbul, Turkey \\
scoleri@ku.edu.tr}
}
\maketitle

\begin{abstract}
In this paper, we consider a full duplex wireless powered communication network where multiple users with radio frequency energy harvesting capability communicate to an energy broadcasting hybrid access point. We investigate the minimum length scheduling and sum throughput maximization problems considering on-off transmission scheme in which users either transmit at a constant power or remain silent. For minimum length scheduling problem, we propose a polynomial-time optimal scheduling algorithm. For sum throughput maximization, we first derive the characteristics of an optimal schedule and then to avoid intractable complexity, we propose a polynomial-time heuristic algorithm which is illustrated to perform nearly optimal through numerical analysis.
\end{abstract}


\section{Introduction} \label{sec:introduction}

The lifetime of a wireless sensor network is usually battery dependent requiring replacement or recharging while the former is either very difficult or infeasible. Recently, radio frequency (RF) energy harvesting arises as the most suitable technology to provide perpetual energy eliminating the need to replace batteries due to design of highly efficient RF energy harvesting hardware \cite{RF}. In wireless powered communication networks (WPCN), wireless users with RF energy harvesting capability; i.e., sensors and machine type communication (MTC) devices, communicate to a hybrid access point (HAP) in the uplink using the energy transferred by the HAP in the downlink  \cite{WPCN_HTT}.

The sum throughput maximization (STM) and minimum length scheduling (MLS) problems have been studied for WPCNs under various models and assumptions. In the half-duplex WPCN models, the users transmit information and harvest energy in non-overlapping time intervals. For half-duplex models, several studies such as \cite{WPCN_commonThroughputmax,WPCN_minThroughputmax,WPCN_weightedThroughputmax} have cosidered the WPCN for common, minimum and weighted throughput maximization respectively. Whereas, \cite{HD-TM,I_pehlwan,Elif_PIMRC} have considered single and multi-antenna WPCN systems for MLS problems. On the other hand, WPCN studies have recently incorporated full duplex technology allowing the access point and the users to achieve simultaneous energy transfer and data communication. Self-interference is the major setback for full-duplex transmission, however, due to recent advances in self-interference cancellation techniques and their practical implementations under the development of 5G and beyond networks, full-duplex has became realizable. The authors in \cite{FD_WPCN_FD_transceiver,FD_WPCN_HD,FD_WPCN_Kang} have considered the full duplex models for MLS and STM. Due to full duplex, a wireless user can harvest energy during both its own and other users transmission, making scheduling important which is missing in these works. Whereas, only few studies \cite{WPCN_scheduling_Kalpant, WPCN_scheduling_Jie, DR_shahid} have paid attention to scheduling but in either a limited context or employing a computationally-inefficient technique. In \cite{WPCN_scheduling_Kalpant}, the scheduling frame is divided into a fixed number of equal length time slots resulting in underutilization of the resources. In \cite{WPCN_scheduling_Jie}, authors have used Hungarian algorithm to find the schedule which is computationally very complex for such sequence dependent transmissions, whereas, \cite{DR_shahid} considered the discrete rate optimization problem. Moreover, these studies have considered simplistic models compared to the system model discussed in this paper. Due to low processing cost and use of simple and cheap power amplifiers, on-off transmission scheme can be very useful for inexpensive sensor networks leading to affordable and widespread deployments of IoT applications. However, in the context of WPCN, no previous study have considered this scheme except \cite{ON-OFF2,ON-OFF1} where the authors have analysed the average error rate and outage probability for a single user system. In this paper, we incorporate on-off transmission scheme in which the users either transmit with a constant power or remain silent if the user can not afford transmission at this power.

The goal of this paper is to revisit MLS and STM problems for determining the optimal time allocation and scheduling considering on-off transmission scheme and a realistic non-linear energy harvesting model in a full-duplex WPCN. The main contributions are listed as follows:

\begin{itemize}
\item We characterize the Minimum Length Scheduling Problem (MLSP) and Sum Throughput Maximization Problem (STMP) and mathematically formulate each as a mixed integer linear programming (MILP) problem.
\item For MLSP, we propose an optimal polynomial-time algorithm incorporating optimal time allocation and scheduling policies.
\item For STMP, upon analyzing the optimality conditions on the optimization variables, we propose a polynomial-time heuristic algorithm and illustrate that it performs nearly optimal for various simulation scenarios.
\end{itemize}
The rest of the paper is organized as follows. The system model and the related assumptions are described in Section \ref{sec:system_model}. The formulation of minimum length scheduling problem and proposed optimal algorithm is presented in section \ref{sec:MLSP}. Section \ref{sec:STMP} covers the sum throughput maximization problem formulation and the proposed scheduling algorithm for the formulated problem. The numerical results are provided  in Section \ref{sec:performance} and concluding remarks are discussed in the Section \ref{sec:conclusion}.

\section{System Model and Assumptions} \label{sec:system_model}

We describe the WPCN architecture and the assumptions used throughout the paper as follows:

\begin{figure}[t]
\centering
\includegraphics[width= 0.8 \linewidth]{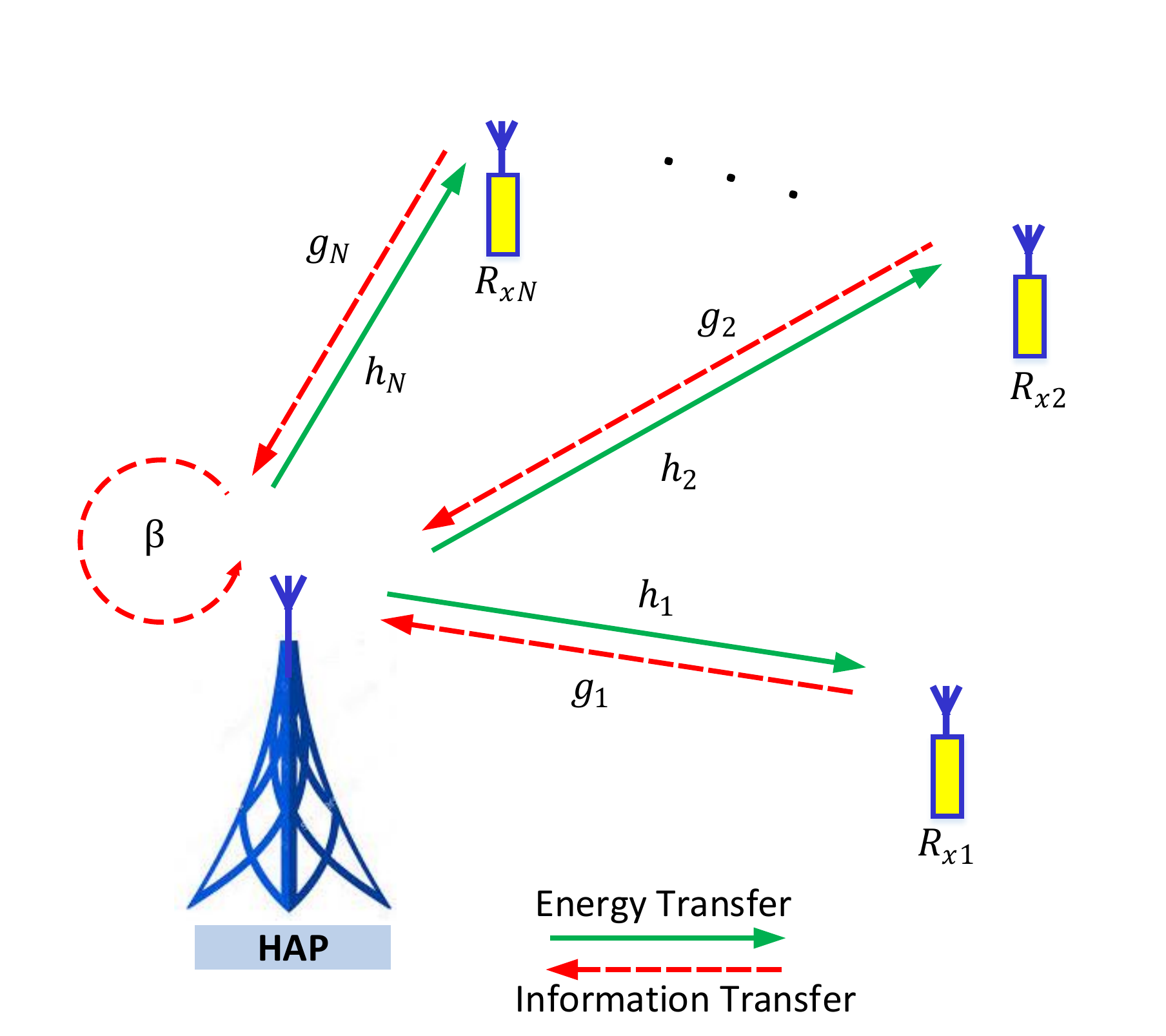}
\caption{Architecture for Wireless Powered Communication Network} \label{fig:WPCN}
\end{figure}

\begin{enumerate}
\item The WPCN architecture, as depicted in Fig. \ref{fig:WPCN}, consists of a HAP and N users; i.e., machine type communications devices and sensors. The HAP and the users are equipped with one full-duplex antenna for simultaneous wireless energy transfer and data transmission on downlink and uplink channels, respectively. The uplink channel gain from user $i$ to the HAP and the downlink channel gain from the HAP to user $i$ are denoted by $g_i$ and $h_i$, respectively. 
\item The HAP has a stable energy supply and continuously broadcasts wireless energy with a constant power $P_h$. Each user $i$ harvests the radiated energy from the HAP and stores in a rechargeable battery. Each user has an initial energy $B_i$ stored in its battery at the beginning of the scheduling frame which includes the harvested and unused energy in the previous scheduling frames.
\item A realistic non-linear energy harvesting model is assumed which is based on the logistic function \cite{NLEH_model_02,NLEH_01} due to its close performance to the experimental results proposed in \cite{NLEH_Prac_01,NLEH_Prac_02,NLEH_Prac_03}. For such a non-linear energy harvesting model the energy harvesting rate for user $i$ is given as: 
\begin{equation}
C_i=\dfrac{P_s[\Psi_i-\Omega_i]}{1-\Omega_i}
\end{equation}
Where, $\Omega_i=\dfrac{1}{1+e^{A_iB_i}}$ is a constant to make sure that zero-input will generate zero-output response, $P_s$ is the maximum harvested power in the saturation state and $\Psi_i$ is a logistic function related to user $i$ which is defined by:
\begin{equation}
\Psi_i=\dfrac{1}{1+e^{-A_i(h_iP_h-B)}}
\end{equation}
Where, $A$ and $B$ are the input power and turn-on threshold constants for the non-linear charging rate respectively. For a given energy harvesting circuit, the parameters $P_s$, $A$ and $B$ can be determined by curve fitting.
\item We consider time division multiple access as medium access control for the uplink data transmissions from the users to the HAP. The time is partitioned into scheduling frames which are further divided into variable-length time slots each of which is allocated to a particular user. 
\item We use constant power model in which all users have a constant transmit power $P_{max}$ during their data transmissions which is imposed to limit the interference to nearby systems.
\item We use constant rate transmission model, in which Shannon capacity formulation for an AWGN channel is used in the calculation of transmission rate $r_i$ of user $i$ as
 \begin{equation} \label{transmission_rate}
 r_i = W \log_{2}(1+k_i P_{max}),
 \end{equation}
where $W$ is the channel bandwidth and $k_i$ is defined as $g_i/(N_oW+\beta P_h)$. The term $\beta P_h$ is the self interference at the HAP and $N_o$ is the noise power density.
\end{enumerate}

\section{Minimum Length Scheduling Problem} \label{sec:MLSP}

In this section, we introduce the minimum length scheduling problem, denoted by MLSP. The joint optimization of the time allocation and scheduling with the objective of minimizing the schedule length is formulated as follows:\\

\textbf{MLSP}:
\begin{subequations} \label{opt_MLSP}
\begin{align}
&\textit{minimize}
& & \sum_{i=0}^{N}  \tau_i  \label{obj_MLSP}\\
& \textit{subject to}
&& W \tau_i \log_2(1+k_iP_{max}) \geq D_i, \label{traffic_MLSP}\\
&&& B_i+C_i ( \tau_0+\sum_{j=1}^{N}a_{ji}\tau_j +\tau_i )-P_{max}\tau_i \geq 0,  \label{EH_MLSP}\\
&&& a_{ij}+a_{ji}=1, \label{ordering_MLSP}\\
& \textit{variables}
& & \tau_i\geq 0, \hspace*{0.1cm} a_{ij} \in \{0,1\}.\label{vars}
\end{align}
\end{subequations}

The variables of the problem are $\tau_i$, the transmission time of user $i$, and $a_{ij}$, a binary variable that takes value $1$ if user $i$ is scheduled before user $j$ and $0$ otherwise. In addition, $\tau_0$ denotes an initial unallocated time in which all users harvest energy without transmitting data.
The objective of the problem is to minimize the schedule length which is equal to the completion time of the transmissions of all users, as given by Eq. (\ref{obj_MLSP}). Eq. (\ref{traffic_MLSP}) represents the traffic requirements of the users where $D_i$ is the amount of data that should be transmitted by user $i$. Energy causality constraint is given by Eq. (\ref{EH_MLSP}): The energy consumed during data transmission cannot exceed the total amount of available energy including both the initial battery level and the harvested energy until and during the transmission of a user. Eq. (\ref{ordering_MLSP}) represents the scheduling order constraint.

In the following, we investigate the characteristics of an optimal solution for MLSP.

\begin{lemma} \label{lemma:time}
There exists an optimal solution of MLSP in which the traffic requirement constraint (\ref{traffic_MLSP}) is satisfied with equality; i.e., each user $i$ transmits exactly $D_i$ bits in the scheduling frame.
\end{lemma}
\begin{proof}
Suppose that $\tau^*=[\tau_0^*,\tau_1^*,\tau_2^*,...,\tau_N^*]$ is the optimal transmission time for a set of users. Then, the optimal schedule length is given by $\sum_{i=0}^{N} \tau_i^*$. Further suppose that, for a user $j$,
\begin{equation} \label{eq:min_time}
\tau_j^* > \tau_j^{min} = \dfrac{D_j}{Wlog_2(1+k_j P_{max})} 
\end{equation}
such that it transmits more than its traffic requirement $D_j$, where $\tau_j^{min}$ denotes the minimum transmission time required by user $j$ to fulfill its traffic requirement. Let $\Delta \tau_j = \tau_j^* - \tau_j^{min}$. The optimal schedule can be updated as $\tau_0^*= \tau_0^*+ \Delta \tau_j$ and $\tau_j^*= \tau_j^*- \Delta \tau_j$. Then, the schedule length does not change while the energy causality requirement of the users are not violated since the completion time of users $i< j$ increase and the completion time of users $i\geq j$ remain same. Therefore, there exists an optimal solution in which  the traffic requirement constraint (\ref{traffic_MLSP}) is satisfied with equality for all users.
\end{proof}

Based on Lemma \ref{lemma:time}, the required energy for the transmission of a user $i$ is given as
\begin{equation}
E_i=\tau_i^{min} P_{max}=\dfrac{D_i P_{max}}{Wlog_2(1+k_iP_{max})}
\end{equation}
Let $s_i^{min}$ denote the minimum starting time for a user $i$ such that it can harvest enough energy to complete its transmission. Then, considering the energy causality constraint (\ref{EH_MLSP}), $s_i^{min}$ is given by
\begin{equation}
s_i^{min}= \frac{E_i - B_i - \tau_i^{min} C_i}{C_i} 
\end{equation}

\begin{lemma}
In an optimal solution of MLSP, users are allocated in increasing order of minimum starting time values.
\end{lemma}
\begin{proof}
Suppose that $\tau^*=[\tau_0^*,\tau_1^*,\tau_2^*,...,\tau_N^*]$ is the optimal transmission time for a set of users. Let $s^*=[s_1^*,s_2^*,...,s_N^*]$ denote the starting time of the users such that $s_1^*<s_2^*<...<s_N^*$ and $s_i^*\geq s_i^{min}$ for all $i\in[1,N]$. Further suppose that for two successively allocated users $j$ and $j+1$, $s_j^{min}>s_{j+1}^{min}$. Hence, $s_{j+1}^* > s_{j}^*\geq s_j^{min}>s_{j+1}^{min}$. The optimal schedule can be updated by interchanging the transmission order of users $j$ and $j+1$ such that user $j+1$ is scheduled at starting time $s_{j+1}=s_{j}^*$ and user $j$ is scheduled just after user $j+1$ completes its transmission at $s_{j}^*+\tau_{j+1}^*$; i.e., $s_{j}=s_{j}^*+\tau_{j+1}^*$. Since $s_{j+1}=s_{j}^*>s_{j+1}^{min}$ and $s_{j}>s_{j}^*\geq s_j^{min}$, both users satisfy their energy causality requirements. 
\end{proof}

The foregoing lemma suggests that at any particular time instant $t$, it is an optimal policy to schedule any user $i$ with $s_i^{min}-t$ is nonpositive and minimum among all $i\in[1,N]$. Then, the optimal schedule should start with an initial unallocated time $\tau_0=\min_{i\in[1,N]} s_i^{min}$ and schedule the user with minimum $s_i^{min}$. Then, it needs to schedule all users in increasing order of $s_i^{min}$ values. Based on the foregoing discussion, we next introduce the Minimum Length Scheduling Algorithm (MLSA), given in Algorithm \ref{algo_MLSA}, for MLSP. 

Input of MLSA algorithm is a set of users, denoted by $\cal{F}$, with the characteristics specified in Section \ref{sec:system_model}. It starts by initializing the schedule $\cal{S}$ where the $i^{th}$ element of $\cal{S}$ is the index of the user scheduled in the $i^{th}$ time slot and the schedule length $t(\cal{S})$. At each step, MLSA picks the user with minimum $s_i^{min}$ value among the unscheduled users. Then, the next time slot is allocated to this user at earliest possible time instant by updating $\tau_0$ accordingly. MLSA terminates when all users in $\cal{F}$ are scheduled and outputs schedule $\cal{S}$ and corresponding set of transmission times $\tau$ including minimum possible $\tau_0$ value required for the successive and continuous transmissions of the users in $\cal{F}$. The computational complexity of MLSA is $\mathcal{O}(N^2)$.

\begin{algorithm} 
\caption{Minimum Length Scheduling Algorithm}  \label{algo_MLSA}
\begin{algorithmic}[1] 
\STATE \textbf{input:} set of users $\cal{F}$
\STATE \textbf{output:} schedule $\cal{S}$, set of transmission times $\tau$, schedule length $t(\cal{S})$
\STATE $\cal{S}$ $\leftarrow$ $\emptyset$, $t(\cal{S})$ $\leftarrow$ 0, $\tau_0$ $\leftarrow$ 0,
\WHILE {$\textbf{F} \neq \emptyset$}
\STATE $m \leftarrow$ arg$\min_{i\in \cal{F}} s_i^{min}$,
\STATE $\cal{S}$ $\leftarrow$ $\cal{S}$ + $\lbrace m\rbrace$,
\STATE $\cal{F}$ $\leftarrow$ $\cal{F}$ - $\lbrace m \rbrace$,
\STATE $\tau_m$ $\leftarrow$ $D_m/(Wlog_2(1+k_mP_{max}))$,
\STATE $\tau_m^{waiting}$ $\leftarrow$ $\max\left\lbrace 0, s_m^{min}-t(\cal{S})\right\rbrace $,
\STATE $\tau_{0}$ $\leftarrow$ $\tau_{0}+\tau_i^{waiting}$,
\STATE $t(\cal{S})$ $\leftarrow$ $t(\cal{S})$ $+$ $\tau_m+\tau_m^{waiting}$,
\ENDWHILE
\end{algorithmic}
\end{algorithm} 

\section{Sum Throughput Maximization Problem} \label{sec:STMP}

In this section, we introduce the sum throughput maximization problem, denoted by STMP. The joint optimization of the time allocation and scheduling is formulated as follows:\\

\textbf{STMP}:
\begin{subequations} \label{opt_STMP}
\begin{align}
&\textit{maximize}
& & \sum_{i=1}^{N}  \tau_i W \log_2(1+k_i P_{max}) \label{obj_STMP}\\
& \textit{subject to}
&& \sum_{i=0}^{N}\tau_i\leq 1, \label{length_STMP}\\
&&& B_i+C_i ( \tau_0+\sum_{j=1}^{N}a_{ji}\tau_j +\tau_i )-P_{max}\tau_i \geq 0,  \label{EH_STMP}\\
&&& a_{ij}+a_{ji}=1, \label{ordering_STMP}\\
& \textit{variables}
& & \tau_i\geq 0, \hspace*{0.1cm} a_{ij} \in \{0,1\}.\label{vars}
\end{align}
\end{subequations}

The objective of the problem is to maximize the sum of the throughput of the users, as given by Eq. (\ref{obj_STMP}). Similar to MLSP formulation, STMP formulation includes the energy causality and scheduling order constraints given by Eqs. (\ref{EH_STMP}) and (\ref{ordering_STMP}), respectively. Besides, the problem formulation employs a normalized schedule length of $1$, as given by Eq. (\ref{length_STMP}), without loss of generality.

STMP formulation is a mixed integer linear programming (MILP) problem thus difficult to solve for the global optimum \cite{opt_book}. On the other hand, for a predetermined transmission order of the users, i.e., $a_{ij}$ values are given, STMP problem is a convex problem for which there exists polynomial-time solution algorithms. Hence, a straightforward solution to STMP would be to enumerate all possible transmission orders, solve each of them and determine the one yielding maximum throughput. However, since there are $N!$ possible transmission orders, such an optimal solution method is intractable. In the following, we present a polynomial-time heuristic algorithm by investigating the characteristics of an optimal solution.

In the following lemma, we present an optimality condition on scheduling suggesting a prioritization among users based on their transmission rates.

\begin{lemma} \label{lemma:rate1}
In the optimal solution of STMP, for any two users $i$ and $j$ such that $r_i>r_j$, if $\tau_i = 0$, then $\tau_j = 0$.
\end{lemma}

\begin{proof}
Suppose that $\tau^*=[\tau_1^*,\tau_2^*,...,\tau_N^*]$ is the optimal transmission time for a set of users such that  $\tau_i^*=0$ and $\tau_j^*>0$ for some $i$ and $j$ such that $r_i>r_j$. For some $\tau^{'} >0$ which will not violate the energy causality requirement of user $i$, transmission time of user $j$ can be divided into two slots of lengths  $\tau_j^*-\tau^{'}$ and $\tau^{'}$, each allocated to users $j$ and $i$, respectively. Then, sum throughput is increased by $\tau^{'} (r_i-r_j)$ which is strictly positive. This is a contradiction.
\end{proof}

While Lemma \ref{lemma:rate1} indicates that high rate users should be prioritized for sum throughput maximization, an optimal schedule does not necessarily contain all users as long as maximum throughput is achieved using a subset of users. However, the following corollary of Lemma \ref{lemma:rate1} states that the maximum rate user should be given a nonzero transmission time.

\begin{corollary} \label{cor:1}
Let user $m$ has $r_m=\max_i r_i$. Then, in an optimal solution, $\tau_m>0$.
\end{corollary}

Moreover, if the maximum rate user has enough initial battery level to transmit with $P_{max}$ during the entire scheduling frame, then, it needs to be allocated to the entire scheduling frame in the optimal schedule. Next, based on the foregoing analysis, we propose the Max-Rate First Scheduling Algorithm (MRSA), given in Algorithm \ref{algo_MRSA}.  Input of MRSA algorithm is a set of users, denoted by $\cal{F}$, sorted in decreasing order of transmission rates. It starts by initializing the unallocated time duration to $1$. At each step, MRSA picks the user with maximum rate and determines the maximum feasible transmission time it can allocate to that user. MRSA performs allocation starting from the end of the scheduling frame to allow higher rate users to harvest more energy. Then, it updates the unallocated time duration accordingly and continues with the next user. If the unallocated time duration is $0$ at any step, MRSA terminates by not scheduling the remaining users. Otherwise it schedules all users and the remaining unallocated time is specified as $\tau_0$. Upon termination, MLSA outputs the schedule $\cal{S}$ consisting of the allocated users and the corresponding sum throughput $R(\cal{S})$. The computational complexity of MRSA is $\mathcal{O}(N)$.

\begin{algorithm}
\caption{Max-Rate First Scheduling Algorithm}  \label{algo_MRSA}
\begin{algorithmic}[1] 
\STATE \textbf{input:} set of users $\cal{F}$ sorted in decreasing order of rates
\STATE \textbf{output:} schedule $\cal{S}$, set of transmission times $\tau$, sum throughput $R(\cal{S})$
\STATE $t^u \leftarrow 1$,
\FOR {$i=1:{|\cal{F}|}$}
\STATE $E_i \leftarrow B_i+C_i t^u$,
\STATE $\tau_i \leftarrow min\lbrace E_i/P_{max}, t^u \rbrace$,
\STATE $t^u \leftarrow t^u - \tau_i $,
\IF {$t^u=0$}
\STATE break,
\ENDIF
\ENDFOR
\STATE $\tau_0 \leftarrow t^u$,
\STATE $\cal{S}$ $\leftarrow \lbrace1,2,...,i\rbrace$
\STATE $R(\cal{S})$ $\leftarrow \sum_{n=1}^i \tau_n r_n$
\end{algorithmic}
\end{algorithm}

\section{Performance Evaluation} \label{sec:performance}

\begin{figure*}[t]
\centering
\includegraphics[width= \linewidth]{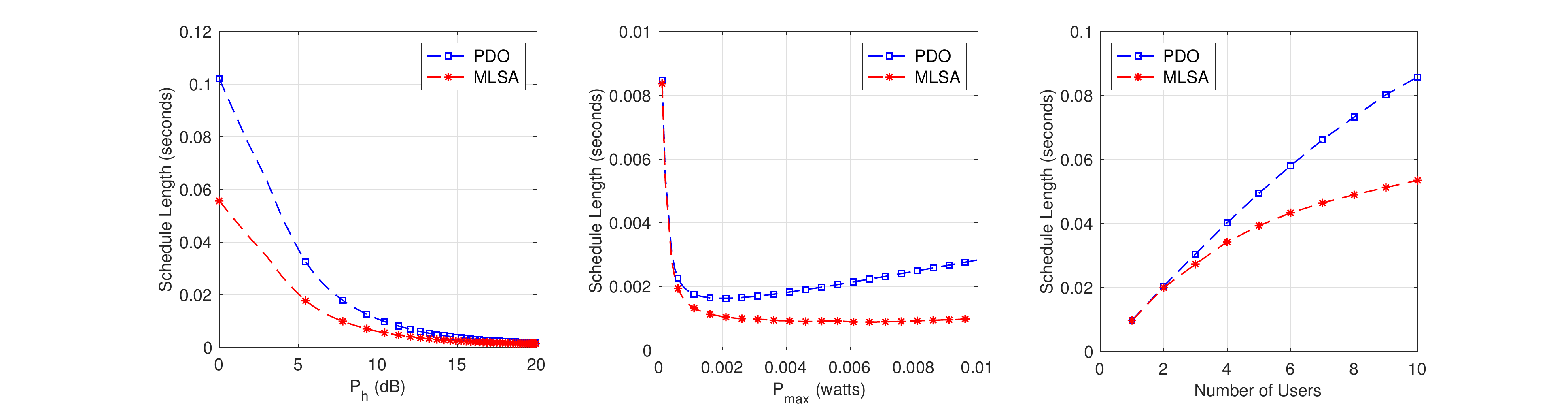}
\caption{Minimum Length Scheduling Performance} \label{fig:MLS}
\end{figure*}

\begin{figure*}[t]
\centering
\includegraphics[width= \linewidth]{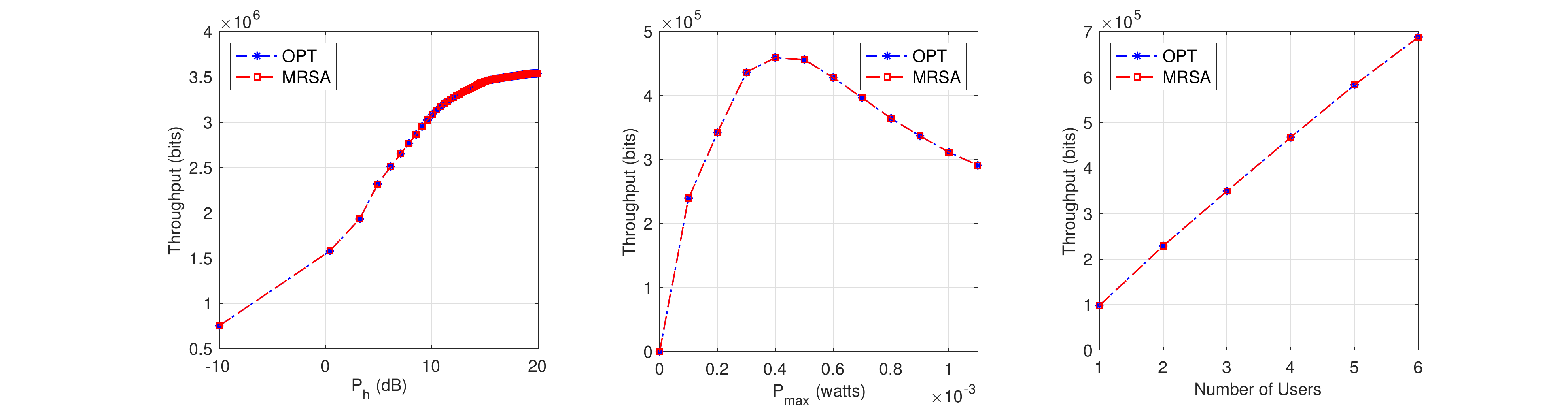}
\caption{Sum Throughput Maximization Performance} \label{fig:STM}
\end{figure*}

The goal of this section is to evaluate the performance of the proposed algorithms. Simulation results are obtained by averaging $1000$ independent random network realizations. The users are uniformly distributed within a circle with radius of $10$m. The attenuation of the links considering large-scale statistics are determined using the path loss model given by 
\begin{equation}
PL(d)=PL(d_0)+10\alpha log_{10}\bigg(\frac{d}{d_0}\bigg)+\emph{Z}
\end{equation}

where $PL(d)$ is the path loss at distance $d$, $d_0$ is the reference distance, $\alpha$ is the path loss exponent, and $Z$ is a zero-mean Gaussian random variable with standard deviation $\sigma$. The small-scale fading has been modeled by using Rayleigh fading with scale parameter  $\Omega$ set to mean power level obtained from the large-scale path loss model. The parameters used in the simulations are $\eta_i=1$ for $i \in [1,N]$; $D_i=100$ bits for $i \in [1,N]$; $W= 1$ MHz; $d_0=1$ m; $PL(d_0)=30$ dB; $\alpha=2.76$, $\sigma=4$ \cite{harvest_50}. The self interference coefficient $\beta$ is taken as $-70$ dBm.

\subsection{Minimum Length Scheduling}

In Fig. \ref{fig:MLS}, we illustrate the performance of the proposed optimal algorithm MLSA in comparison to a predetermined scheduling order based algorithm, denoted by PDO, for different scenarios. PDO simply allocates the users in a given arbitrary order and thus does not exploit the benefit of optimal scheduling to decrease the length of the scheduling frame. We first illustrate the scheduling performance for different $P_h$ values. Schedule length decreases with the increasing $P_h$ since higher HAP power indicates that users can start and thus complete their transmissions earlier in the scheduling frame since any user will be able to afford $P_{max}$ transmit power earlier via harvesting more energy from the HAP. While for large values, optimal scheduling loses its importance on the performance, for relatively small and practical values of $P_h$, MLSA outperforms PDO significantly. A similar superiority of MLSA can be observed from the figure for increasing $P_{max}$ values. As $P_{max}$ increases, performance of both algorithms initially improve since users continue to afford $P_{max}$ transmit power using their initial battery levels at the very beginning of the scheduling frame. However, above a critical value of $P_{max}$, increasing transmit power leads to increasing initial unallocated time $\tau_0$ in the scheduling frame. This results in a performance degradation for PDO while MLSA accomodates this effect by optimally determining the scheduling order. Finally, MLSA outperforms PDO for increasing number of users in the WPCN. While the schedule length almost increases linearly for PDO as the number of users increases, the increase in the schedule length diminishes for MLSA again indicating the significance of determining optimal transmission order.

\subsection{Sum Throughput Maximization}

In Fig. \ref{fig:STM}, we illustrate the throughput performance of the proposed algorithm MRSA in comparison to the optimal solution, denoted by OPT. Optimal solution is obtained by enumerating all possible transmission orders and picking the one yielding the highest throughput via solving a convex optimization problem for each transmission order. One clear observation is that MRSA performs nearly optimal on average while achieving exact optimal solutions in most network realizations. As HAP power $P_h$ increases, sum throughput yielded by MRSA increases while it saturates for large values of $P_h$ since the energy that can be used by the users in a scheduling frame is limited. For increasing $P_{max}$, sum throughput first increases since users can have higher transmission rates. Then, above certain $P_{max}$ values, users cannot afford $P_{max}$ in the beginning of the scheduling frame resulting an increase in the initial unallocated time and thus decrease in the total allocated time by the users. As the number of users increases, sum throughput achieved by the users almost increases linearly as expected ideally.


\section{Concluding Remarks} \label{sec:conclusion}

In this paper, we have investigated minimum length scheduling and sum throughput maximization problems for a full duplex WPCN considering on-off transmission scheme. For both problems, we have derived the characteristics of the optimal solution and proposed polynomial-time solution schemes. As future work, we plan extending this study for discrete rate based transmission rate model in which users can select a transmission rate from a finite set based on their SNR levels. Moreover, the WPCN architecture for multiple hybrid access points will also be investigated.

\bibliography{shahid_bib}
\bibliographystyle{ieeetr}
\end{document}